\documentclass[aps, pra, a4paper, showpacs, twocolumn, english, 10pt, nofootinbib]{revtex4-2}
\usepackage{bbm, amsthm, bm,textcomp, nicefrac,geometry,ragged2e}
\usepackage[dvipsnames]{xcolor}
\usepackage{float}
\usepackage[bbgreekl]{mathbbol}
\usepackage{graphicx,epstopdf,color,verbatim,enumitem}
\usepackage[normalem]{ulem}
\usepackage{pbox,array}
\usepackage{booktabs}
\usepackage{mathrsfs}
\usepackage{amssymb}
\usepackage{textgreek}
\usepackage{mathtools}
\usepackage{stackrel}
\usepackage[thinlines]{easytable}
\usepackage{amsmath}
\usepackage[utf8]{inputenc}
\makeatletter
\usepackage{soul}
\usepackage[caption=false]{subfig}
\usepackage{hyperref}
\hypersetup{
    colorlinks = true,
    linkcolor = teal,
    citecolor = teal,
    filecolor = Black,
    urlcolor = Black,
}
\usepackage[T1]{fontenc}
\usepackage[caption=false]{subfig}
\usepackage{babel}
\usepackage[linesnumbered,ruled]{algorithm2e}

\addto\captionsenglish{}

\newtheorem{definition}{Definition}
\newtheorem{theorem}{Theorem}

\definecolor{brickred}{rgb}{0.8, 0.0, 0.0}

\usepackage{tikz}
\usetikzlibrary{shapes.geometric, positioning, decorations.pathmorphing, shadows, calc}

\definecolor{nodeblue}{RGB}{44, 62, 80}       
\definecolor{activeorange}{RGB}{230, 126, 34} 
\definecolor{linkgray}{RGB}{149, 165, 166}    

\begin{document}

\title{Universal Nested Quantum Switch}

\author{Jorge Miguel-Ramiro, Maria Flors Mor-Ruiz and Wolfgang D\"ur$^1$}

\affiliation{$^1$Universit\"at Innsbruck, Institut f\"ur Theoretische Physik, Technikerstra{\ss}e 21a, 6020 Innsbruck, Austria}

\date{\today}

\begin{abstract}
The quantum switch is a basic network primitive that allows one to connect multiple nodes in a quantum network via a central node. We show that the same functionality can be achieved with a different geometry that does not rely on a powerful and large central unit, but instead utilizes evenly distributed resources. This approach is resilient against node failures. We provide a nested construction with logarithmically many qubits per node and a total of $O(n\log n)$ Bell pairs, in contrast to other distributed approaches based on pre-shared entanglement that scale as $O(n^2)$. The construction achieves fully flexible pairwise connectivity, where the shared resource state can be locally transformed into $n/2$ arbitrarily distributed Bell states. We also present a graph state variant with just one qubit per node, which allows one to generate $O(n/\log^2 n)$ Bell pairs.
\end{abstract}

\maketitle

\section{Introduction}

Quantum networks promise to enable fundamentally new communication and quantum information processing capabilities by exploiting the distribution and manipulation of entanglement across spatially separated nodes \cite{Kimble2008, Wehner2018, Kozlowski2019, MiguelRamiro2021, Wei2022, Azuma2023}. Applications range from quantum key distribution \cite{Mehic2020, Wolf2021, Cao2022} and distributed quantum sensing  \cite{Sekatski2020, Zhang2021, Bugalho2025} to distributed and blind quantum computing \cite{Cacciapuoti2020, Main2025, Fitzsimons2017, Barz2012}. As networks scale in size and complexity, a central challenge is how to establish flexible, parallel entanglement connections between many nodes using limited and realistic quantum resources.

A basic network primitive addressing this challenge is the \textit{quantum switch}, whose role is to dynamically connect nodes upon request, potentially in parallel \cite{Vardoyan2023, Vardoyan2021, Dai2021, Avis2023, Gauthier2023, PrezCastro2024, MiguelRamiro2023}. Conceptually, a quantum switch enables the transformation of shared resources into a desired set of bipartite Bell pairs (or multipartite states), allowing different communication patterns to be realized. Such functionality can be implemented either by generating entanglement on demand using physical channels, or by consuming pre-distributed entangled resources.

Existing proposals of quantum switch functionality typically fall into two broad categories. Centralized approaches \cite{ Vardoyan2021, Panigrahy2023, Vardoyan2023, Dai2021, Avis2023, Gauthier2023,  Lee2022} rely on a powerful hub node that shares entanglement with all other nodes and performs joint measurements to establish requested connections. While this strategy is resource-efficient in terms of total entanglement and storage, it requires a large and highly reliable central memory presenting a single point of critical failure. At the opposite extreme, fully decentralized schemes \cite{MiguelRamiro2023} based on pre-sharing Bell pairs between all node pairs provide maximal flexibility, but at the cost of quadratic scaling in both total entanglement and local memory, making them impractical beyond small network sizes. Intermediate proposals optimize storage by tailoring resources to specific traffic patterns or network topologies, but generally do not provide relevant scaling benefits with respect to Bell-pair approaches.

In this work we introduce an alternative approach to quantum switching based on pre-existing entanglement that is fully distributed and avoids centralized control. Our construction, which we refer to as a \textit{nested quantum switch}, employs a highly symmetric entanglement geometry in which resources are distributed across the network. Each node stores only a logarithmic number of qubits, yet the global resource can be locally transformed to realize \textit{any} simultaneous bipartite connectivity pattern between the nodes. In this sense, the nested quantum switch achieves the functionality of a universal quantum switch while significantly improving the scaling of local quantum memory compared to other decentralized proposals.

Beyond favorable scaling, the distributed and symmetric nature of the resource leads to inherent resilience against failures. Since connectivity is not mediated by a single distinguished node, the loss of nodes or entangled links only involves a gradual degradation of performance rather than catastrophic failure. We analyze this robustness, showing that a large fraction of connectivity requests remain feasible even under losses or node failures. We analytically show that the nested architecture preserves high quality entanglement, as the end-to-end fidelity exhibits only a polynomial decay with the number of nodes.

Finally, we consider an extreme low-memory variant in which the entanglement structure is merged such that each node holds only a single qubit. While this setting no longer allows full flexibility, we show that it still supports the generation of a scalable number of simultaneous Bell pairs, highlighting a smooth trade-off between local memory and achievable parallel connectivity.

The remainder of this paper is organized as follows. In Sec~\ref{sec:back} we review the concept of a quantum switch and the existing strategies to realize it.  In Sec.~\ref{sec:model} we introduce the network model and the nested entanglement resource. We establish and prove the conditions under which universal bipartite connectivity can be achieved.  We also present numerical studies of robustness under node and link failures.  We provide numerical robustness and analytical fidelity analysis in Sec. \ref{sec:performance}. In Sec.~\ref{sec:graphstate} we discuss a low-memory variant of the construction. We conclude in Sec.~\ref{sec:discussion} with a discussion of implications for scalable quantum network architectures.

\begin{figure}
    \centering
    \includegraphics[width=0.85\linewidth]{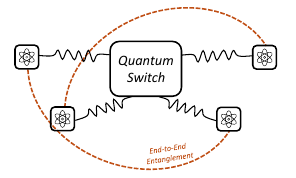}
    \caption{Quantum switch. General idea. A device or network architecture that enables nodes in a quantum network to establish end-to-end entanglement in different connectivity patterns, potentially simultaneously and in parallel.}
    \label{fig:switch_general}
\end{figure}

\section{Quantum switch. Definition and existing proposals } \label{sec:back}
We begin by formally defining the notion of a quantum switch in a network scenario, as distinct meanings of this term appear in the literature. We then review existing proposals that realize quantum switch functionality under different architectural assumptions, before introducing our alternative strategy in Sec.~\ref{sec:model}.

\subsection{Definition of a quantum switch}

In the context of quantum networks, the term \textit{quantum switch} is used with different meanings across the literature. In this work, we adopt a network level definition that is conceptually close to its classical counterpart, see Fig. \ref{fig:switch_general}.

{\definition{(\textbf{Quantum switch}).}
We define a quantum switch as a network primitive that enables the generation of end-to-end entanglement between selected pairs of (or, more generally, groups of) nodes on demand, potentially in parallel. The target request of a switching operation is a set of disjoint bipartite or multipartite states between network nodes, which can subsequently be used by higher-layer applications.}

Operationally, quantum-switch functionality is typically realized by transforming a collection of elementary entangled states, either generated on demand or pre-distributed, into the requested end-to-end connectivity configurations using local operations and classical communication. Throughout this work, we consider settings in which multiple pairwise connections may be established simultaneously, subject to the constraint that quantum resources (e.g., entangled links or local memory qubits) cannot be reused.

This definition is deliberately agnostic to the underlying network architecture. In particular, it encompasses both centralized realizations, where switching operations are mediated by a distinguished node, and fully distributed realizations, where switching capability emerges from the structure of the shared entanglement resource.

{\definition{(\textbf{Pairwise Universality}).}
For a network of $n$ nodes, any connectivity request can contain at most $n/2$ parallel bipartite links between distinct pairs of nodes. If a given quantum switch construction can guarantee \textit{any} configuration consisting of $n/2$ disjoint Bell pairs, corresponding to full pairwise connectivity, then we say the switch has pairwise universality.  }

The number of distinct requests grows rapidly with the network size. Specifically, the number of different ways to pair $n$ nodes into $n/2$ disjoint bipartite connections is given by the number of fixed-point-free involutions, i.e., $(n-1)!! \;=\; \prod_{k=1}^{n/2} (2k-1)$. 
Note that any request involving fewer than $n/2$ simultaneous links can also be obtained by discarding unused pairs.

We restrict in this work to quantum switch constructions that are pairwise universal. For simplicity, we refer to such constructions simply as universal throughout. Extensions to multi-round operation, multi-partite target states, or traffic-dependent scheduling policies are beyond the scope of this work.

\subsubsection{Relation to other paradigms}

In the broader context of quantum networks, the generation of end-to-end entanglement between remote nodes has been studied under several related paradigms, most notably entanglement distribution \cite{Gyongyosi2019, Meignant2019, Fischer2021, Sutcliffe2023, MR2025, Mazza2025} and entanglement routing \cite{Pant2019, Schoute2016, Lee2022, Shi2024, Abane2024, Devulapalli2024}. Entanglement distribution refers to the process of generating and delivering entangled states across a network, typically by creating short-range entangled links and combining them via entanglement swapping at intermediate nodes to span longer distances. Entanglement routing generalizes this notion by addressing how sequences of adjacent entangled links and swapping operations should be selected in order to satisfy one or more end-to-end entanglement requests under probabilistic link generation, fidelity constraints, and finite coherence times. In many settings, routing protocols choose paths through a given network topology in close analogy with classical routing, taking into account network state information and performance objectives such as throughput or latency. 

While entanglement distribution and routing protocols focus on the generation and scheduling of end-to-end entanglement over time, they do not by themselves define a network primitive for realizing arbitrary sets of simultaneous pairing requests. In contrast, a quantum switch, as defined here, is an abstract primitive that makes use of entangled resources to generate a specified set of disjoint end-to-end Bell pairs in parallel. Entanglement routing may be employed as a subroutine within a switch implementation, but the switch abstraction emphasizes on-demand, simultaneous connectivity and instantaneous parallel reconfiguration, rather than sequential path-wise generation of individual entangled pairs. In an informal sense, routing addresses how end-to-end entanglement can be established over time under network constraints, whereas switching focuses on which global connectivity patterns can be realized in parallel within a single network use.

At the same time, we distinguish this notion from the quantum switch introduced in quantum information theory as a higher-order operation controlling the order of quantum channels, which is unrelated to the networking functionality considered here.

\subsection{Existing quantum switch architectures}
We review here known proposed realizations of universal pairwise quantum switches, distinguishing between centralized constructions based on a single central hub node and distributed approaches.

\subsubsection{Centralized approaches}
A common realization of quantum-switch functionality relies on a distinguished central node acting as an entanglement hub, see Fig. \ref{fig:switch_prev} (a). In this star-like topology architecture, each user (or end node) is connected to the hub via a physical link over which link level entangled states (e.g., Bell pairs) are generated. The hub stores one or more qubits per user and performs Bell-state measurements and entanglement swapping operations to establish end-to-end entanglement between arbitrarily selected user pairs on demand \cite{Vardoyan2023, Vardoyan2021, Dai2021, Avis2023, Gauthier2023, PrezCastro2024}. 

Centralized quantum switches have been studied from several complementary perspectives. One line of work models the switch as a queuing system that processes stochastic entanglement requests, generates link-level entanglement probabilistically, and performs swapping according to a scheduling policy; the corresponding capacity regions and stability conditions have been characterized under various assumptions \cite{Vardoyan2019, Vardoyan2023, Vardoyan2021, Panigrahy2023}. Other studies incorporate physical constraints such as finite memory lifetimes and imperfect links  \cite{Kumar2023, Avis2023}. Control and allocation mechanisms for these architectures have also been proposed to ensure  efficient resource utilization \cite{Gauthier2023, Lee2022}.

Centralized strategies are appealing because the number of physical systems scales linearly with the number of users. However, they inherently concentrate quantum memory and operational complexity at a single hub, creating an critical single point of failure. As a result, existing analyses typically emphasize average-case performance metrics under stochastic traffic models, rather than worst-case or combinatorial guarantees for realizing arbitrary simultaneous pairing requests among all users.

\subsubsection{Distributed strategies}
Beyond centralized switch architectures, several distributed strategies have been proposed to realize quantum switch functionality by pre-sharing entangled resources across the network. The most straightforward construction is based on pre-established Bell pairs between all node pairs, see Fig. \ref{fig:switch_prev} (b). A complete Bell-pair resource trivially guarantees universality, in the sense that any simultaneous pairing configuration can be realized. This, however, comes at the cost of quadratic scaling in both total entanglement and local memory. Concretely, the number of required Bell pairs for a network of $n$ nodes is $\sum_{i=1}^{n}(n-i)={n(n-1)}/{2}$, implying $\Theta(n)$ qubits per node.

\begin{figure}[t!]
    \centering
    \includegraphics[width=\linewidth]{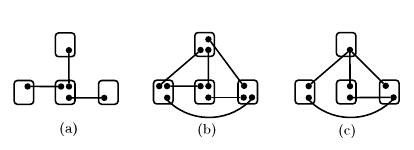}
    \caption{Different proposals for quantum switches.
(a) Centralized approach: a distinguished central node shares Bell pairs with all network nodes and performs entanglement swapping to establish end-to-end connections.
(b) Distributed Bell-pair approach: all-to-all direct connectivity based on  Bell pairs between every node pair.
(c) Decreasing size GHZ approach: multipartite GHZ entangled states of varying sizes are used to enable pairwise connections via local measurements.
All these strategies can realize universal pairwise connectivity. }
    \label{fig:switch_prev}
  
\end{figure}

An alternative class of distributed constructions is based on multipartite entangled resources, such as GHZ states or more general graph states \cite{Hein2004, Hein2006}. In particular, it was shown \cite{Pirker2018, Pirker2019} that collections of GHZ states of decreasing size can be used to realize universal pairwise connectivity via appropriate local measurements, see Fig. \ref{fig:switch_prev} (c). While this approach reduces the need for storing all Bell pairs explicitly, the total entanglement cost still scales quadratically with the network size. Specifically, the total number of entangled qubits required scales as $\sum_{i=0}^{n/2-1}(n-i)=n(3n+2)/8$, 
and the generation and stabilization of multipartite states introduce additional operational complexity compared to bipartite resources.

Other distributed constructions based on structured network topologies have also been considered. For instance, generalizations of the butterfly network \cite{Hahn2019, MiguelRamiro2023} or 2D graph states \cite{Freund2024} were discussed, showing however similar scaling features. 

Overall, all known distributed constructions that guarantee universal pairwise connectivity  require resources scaling quadratically with the number of nodes. In contrast, the nested quantum switch introduced here provides a distributed Bell-pair-based resource that achieves universality with efficient scaling.

\begin{table*}[ht]
\centering

\label{tab:switch_comparison}
\begin{tabular}{@{}llll@{}}
\toprule
\textbf{Architecture} & \textbf{Memory per Node} & \textbf{Total Resource} & \textbf{Robustness} \\ \midrule
Centralized Hub \cite{Vardoyan2023, Vardoyan2021, Dai2021, Avis2023, Gauthier2023, PrezCastro2024} & $O(1)$ (User) / $O(n)$ (Hub) & $O(n)$ & Low (Single point of failure) \\
All-to-All Bell \cite{MiguelRamiro2023} & $\Theta(n)$ & $O(n^2)$ & High (Distributed) \\
GHZ-based \cite{Pirker2018, Pirker2019} & $\Theta(n)$ & $O(n^2)$  & Low (GHZ loss sensitivity) \\
{Nested Switch (Ours)} & $O(\log n)$ & $O(n \log n)$  & {High (Symmetry)} \\ \bottomrule
\end{tabular}
\caption{Comparison of different quantum switch architectures. }
\end{table*}

\section{Nested quantum switch} \label{sec:model}
We now introduce the nested quantum switch, a fully distributed and symmetric entanglement resource that enables universal parallel bipartite connectivity without relying on a central hub node. We first describe the resource construction and its graph-theoretic interpretation, and then formalize the switching problem and state the main universality result.

Consider a network of $n = 2^d$ nodes. Each node stores $2d-1 =(2\log_2 n- 1)$ qubits \footnote{Throughout this work, all logarithms are taken to base 2. For simplicity, we denote $\log_2$ simply by $\log$.}  (due to boundary conditions). The entanglement resource is constructed such that each node shares exactly one Bell pair per qubit with distinct other nodes, with the interaction distances following the hierarchical pattern $(1,\,2,\,4,\,8,\,\dots,\,2^{d-1} )$ in a nested fashion (see Fig.~\ref{fig:switch_general}), i.e., each node  $x$ connects to neighbors at modular distances $x \pm 2^k \pmod{2^d}$. As a result, each node participates in $O(d)$ Bell pairs, and the total number of Bell pairs in the network scales exactly as $O(n\log n)$.

The resulting entanglement structure defines a regular graph in which vertices correspond to network nodes and edges correspond to Bell pairs. This hierarchical entanglement pattern induces a highly symmetric interaction graph that strictly embeds a $d$-dimensional hypercube as a spanning subgraph, as formalized below.

\begin{figure}[t!]
    \centering
    \includegraphics[width=\linewidth]{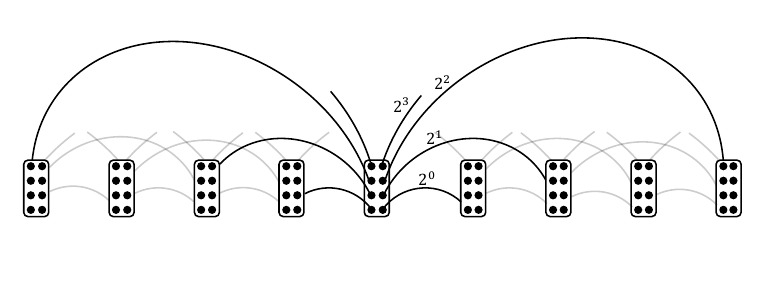}
   \caption{Nested quantum switch.
Each node is connected via Bell pairs to neighbors at distances $2^k$ ($k=0,\dots,d-1$) in a highly symmetric and distributed entanglement pattern.
Universal pairwise connectivity is guaranteed with total resource scaling of just $O(n\log n)$.}
    \label{fig:switch_general}
\end{figure}

\subsection{The nested quantum switch as a hypercube}
A classical $d$-dimensional hypercube $Q_d$ \cite{Harary1988} is a graph with $n = 2^d$ vertices, where each vertex is represented by a unique bitstring of length $d$. Two vertices are connected by an edge if and only if their bitstrings differ by exactly one bit. Formally, for any node $z \in \{0, \dots, n-1\}$, its neighbors are defined by the bitwise XOR operation $z \oplus 2^k$ for $k \in \{0, \dots, d-1\}$. In our network model, an edge between two vertices represents a shared Bell pair, meaning each node in a hypercube architecture stores $d = \log n$ qubits.

Within this graph theory framework, entanglement swapping along a path corresponds to sequentially consuming all Bell pairs to establish an end-to-end connection between the path endpoints. Since each Bell pair is a finite resource that can only be consumed once, simultaneous connections must be implemented using edge-disjoint paths.


We can interpret the $d$-dimensional hypercube as a spanning subgraph of the physical nested architecture. The distinction lies in the boundary conditions: while the hypercube is defined by bit-flip logic ($z \oplus 2^k$), the nested switch is built on cyclic modular arithmetic ($z \pm 2^k \pmod n$). Crucially, flipping the $k$-th bit of a node $z$ corresponds to either adding $2^k$ (if the $k$-th bit is 0) or subtracting $2^k$ (if the $k$-th bit is 1). Consequently, every hypercube edge $z \oplus 2^k$ maps perfectly to exactly one of the two bidirectional physical links ($z+2^k$ or $z-2^k$) present in the nested geometry.

The cyclic symmetry of the nested switch results in a higher connectivity, specifically a degree of $2d-1$ compared to the hypercube degree $d$. While the nested switch is not exactly isomorphic to two copies of a hypercube due to its non-bipartite symmetry, it possesses an equivalent total edge capacity. Specifically, the nested switch provides $nd - n/2$ physical Bell pairs, matching the resource of the two hypercube layers required for 2-rearrangeability \cite{Gu1997, Choi1993, Shen1994}. This relationship ensures the nested switch inherits the hypercube universal routing capacities.

A switching request is specified by a set of disjoint node pairs, $
\{(u_1,v_1),(u_2,v_2),\dots,(u_{n/2},v_{n/2})\}$, where $u_k,v_k \in \{1,\dots,n\}$ label network nodes. Such a set corresponds to a perfect matching \cite{Fink2007, Diestel2025} of the vertex set, i.e., each node appears in exactly one pair. Equivalently, the request can be identified with a fixed-point free involution $\pi = (u_1,v_1)(u_2,v_2)\cdots(u_{n/2},v_{n/2})$,
where each transposition denotes a requested Bell pair.

Given a switching request $\pi$, the task of the quantum switch is to establish Bell pairs between each $u_i$ and $v_i$ simultaneously using only local operations and classical communication, while consuming each Bell pair in the resource at most once. In graph-theoretic terms, this is equivalent to finding, for each transposition $(u_i,v_i)$, a path in the hypercube graph such that all selected paths are pairwise edge-disjoint.


While general permutations may involve cycles of arbitrary length, an involution is a permutation composed only of disjoint transpositions and fixed points. Fixed-point free involutions correspond exactly to perfect matchings of the vertices \cite{Diestel2025} and describe the class of pairwise connection requests considered in this work.


\subsection{Proof of universality of the nested quantum switch}
Our central theoretical observation is that the switching problem defined above coincides with the problem of routing involutions on the hypercube using edge-disjoint paths. This problem has been studied in the classical literature in the context of permutation routing \cite{Sprague1994}.


\begin{theorem}[Universality of the nested quantum switch]
Consider a network of $n = 2^d$ nodes connected via a nested quantum switch. By restricting the entanglement resource strictly to its $d$-dimensional hypercube spanning subgraph, the switch is pairwise universal: for any switching request consisting of $n/2$ disjoint node pairs, there exists a sequence of local operations and entanglement swapping operations that establishes all requested Bell pairs, without reusing any Bell pair within the same resource layer.
\end{theorem}

\begin{proof}[Proof sketch]
The proof follows directly from the equivalence between simultaneous Bell-pair generation and edge-disjoint path routing. As established, the $d$-dimensional hypercube is a spanning subgraph of the nested architecture: any logical bit-flip operation $x \oplus 2^k$ mathematically corresponds to either $x + 2^k$ or $x - 2^k \pmod{n}$, both of which are native physical links in the nested topology. By restricting our routing to this hypercube subgraph, each switching request defines a fixed-point free involution on its vertices. For odd $d$, classical results \cite{Sprague1994} guarantee the existence of automorphic edge-disjoint paths realizing this involution. Performing entanglement swapping along these paths yields the desired Bell pairs while consuming each resource edge at most once. If one allows for two copies of the resource state—ensuring the total resource remains within the $O(n\log n)$ scaling—then arbitrary pairing and permutation requests can be supported for any $d$. This follows from the fact that the hypercube is 2-rearrangeable, meaning that any general permutation can be decomposed into two sub-permutations, each routable by edge-disjoint paths \cite{Gu1997, Choi1993, Shen1994}.
\end{proof}

We note again that the nested architecture, with a node degree of $2d-1$, physically provides nearly the same total resource as the two $d$-regular hypercube layers required for universality.

While constructive algorithms for the specific decomposition exist, e.g. \cite{Gu1997}, they are hard to implement.  We employ a heuristic procedure to further evaluate the switch performance. In the next section, we show via numerical simulations that even a simple algorithm based on multi-path finding applied to the nested switch construction achieves strong performance, essentially universality, and exhibits good resilience under node failures.

\section{Performance analysis} \label{sec:performance}
To obtain a qualitative robustness benchmark, we evaluate a multi-path routing protocol on our nested switch resource, using its $d$-dimensional hypercube spanning subgraph. In each trial, we remove $x$ nodes uniformly at random together with all incident edges. We evaluate performance under maximal load ($(n-x)/2$ random pairs) using a $k$-shortest path heuristic algorithm (where we set $k=20$) on the nested switch. For each request, the algorithm selects the first available path from the $k$ shortest options that satisfies the link capacity constraint $R$ (Bell pairs per link). This multi-path strategy allows the network to overcome local congestion and failures via detours, yielding a realistic lower bound on achievable connectivity. For a given per-edge capacity $R$, we define the served fraction $S_R = {\#\text{served pairs}}/{\#\text{requested pairs}},$ where a request consists of a simultaneous perfect matching on the surviving nodes.

We remark that our physical nested switch construction nearly corresponds to a hypercube with $R=2$.

\subsection{Resource overhead and robustness under losses}
We start analyzing the resource overhead of the protocol in the edge-disjoint setting by measuring the average number of link-level Bell pairs consumed per served or delivered end-to-end pair, see Fig. \ref{fig:fig1}. This quantity is equal to the mean path length of accepted connections. As the number of failed nodes increases, the mean cost changes only mildly, remaining within a narrow band across the explored failure range. This indicates that node loss has little effect on the resource cost of the connections that remain feasible under the routing rule. 

Note that this metric is directly related to the physical quality of the switch, as the number of hops $L$ dictates the number of required entanglement swapping operations. Assuming imperfect generation and local operations, the end-to-end fidelity scales as $F_{\text{total}} \propto F_{\text{swap}}^{L-1}$, where as stressed before $ L=(\log(n))/2$. Consequently, the observation that the path length remains close to the theoretical minimum, even under substantial node failure, indicates that the switch preserves high achievable fidelities, exhibiting strong robustness. See below for a more concrete fidelity analysis.

\begin{figure}
    {\includegraphics[width=0.95\columnwidth]{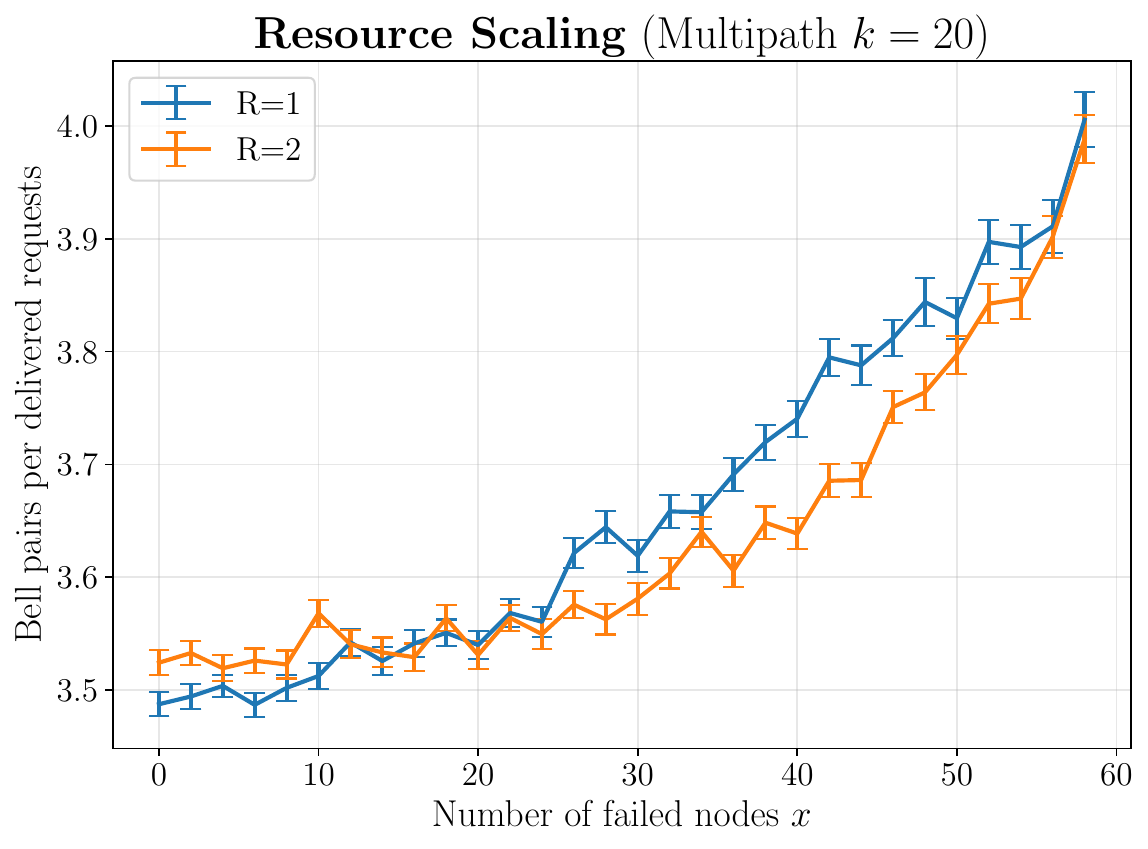}}
    \caption{\label{fig:fig1} Resource scaling. Average number of Bell pairs consumed per delivered connection in a simultaneous target configuration as a function of the number of failed nodes for a network of $128$ nodes. Each data point is averaged over $M$ independent Monte Carlo trials, where each trial samples a random failure pattern and a random perfect matching on the surviving nodes. Error bars denote the standard error deviation. 
}
\end{figure}

Fig.~\ref{fig:fig2} shows the mean fraction of request pairs served by the routing protocol as a function of the number of failed or lost nodes, for several per-edge capacities $R$. Even in the absence of failures ($x=0$), the basic protocol with $R=1$ serves only about $90\%$ of simultaneous requests due to path collisions induced by the fixed routing constraints. As nodes fail, the performance only decreases smoothly. Increasing the per-edge capacity has a strong stabilizing effect: already for $R=2$ (which, as mentioned before, nearly correspond to our nested construction), almost all requests, i.e., pairwise universality, are served across the entire failure range, i.e., linear resource overhead, absorbs random node loss effectively.

\begin{figure}
    {\includegraphics[width=0.95\columnwidth]{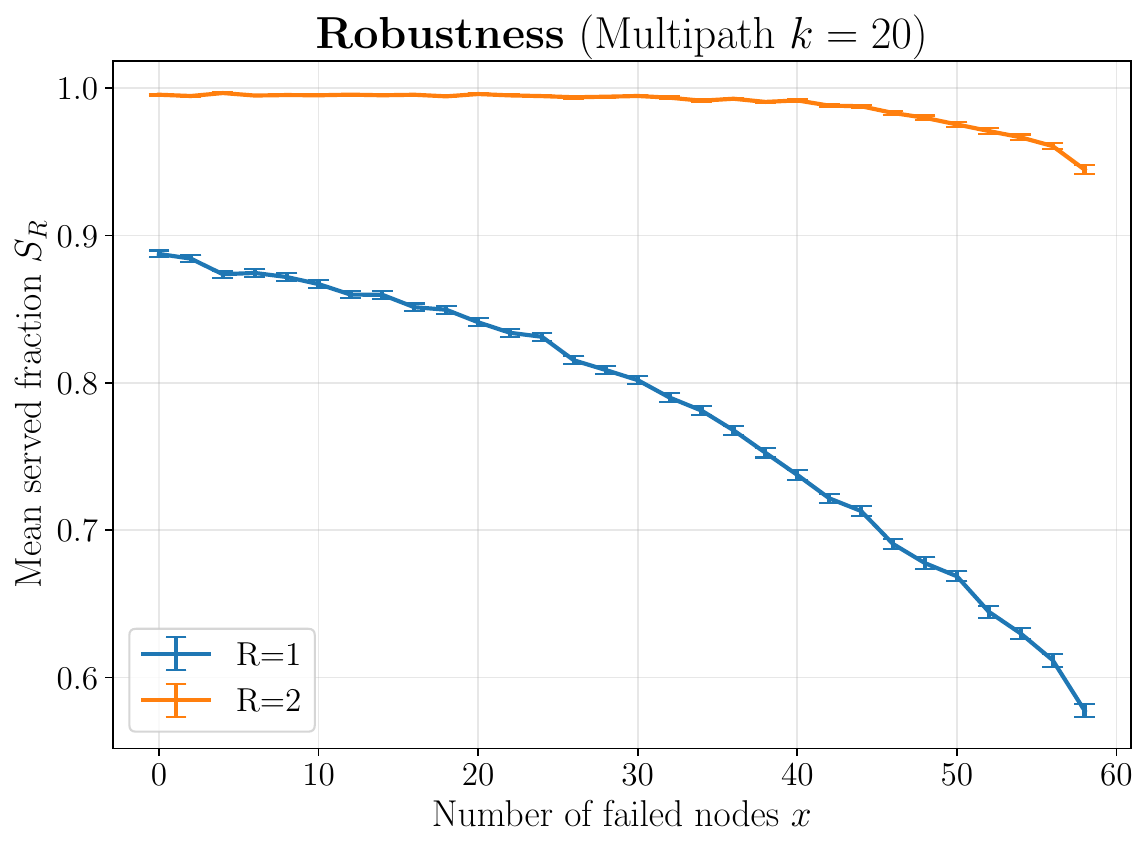}}
    \caption{\label{fig:fig2}  Robustness of the nested quantum switch against node failures. The mean fraction of served requests $S_R$ is shown for a network of $n=128$ nodes as a function of the number of failed nodes $x$. Routing is performed using a greedy $k$-shortest path algorithm with $k=20$, for edge capacities of $R=1$ (blue) and $R=2$ (orange) Bell pairs. We remark that the $R=2$ case corresponds to the actual behavior of the nested construction. Each data point is averaged over $M=500$ independent Monte Carlo trials. Error bars represent the standard error deviation. Since the served fraction is strictly bounded in $[0,1]$, concentration inequalities, e.g., Hoeffding's inequality, imply that the estimation error scales as $O(1/\sqrt{M})$. For $M=500$, the resulting statistical uncertainty is approximately $2\%$.
}
\end{figure}

\subsection{Scalability}
Complementing the previous performance metrics, we examine the distribution of network traffic to verify the absence of structural bottlenecks. Fig. \ref{fig:fig3} shows the probability distribution of edge loads under a maximal switching request ($n/2$ simultaneous pairs), obtained via the multi-path routing protocol. The histogram confirms that the network suffers from no structural bottlenecks. Under full saturation, the traffic is uniformly distributed, with no link requiring more than $R=2$ Bell pairs independently of the network size. This uniform spreading of traffic confirms that the nested switch topology naturally balances the load without creating hotspots or localized congestion points, in contrast to central node approaches. 

\begin{figure}
    {\includegraphics[width=0.95\columnwidth]{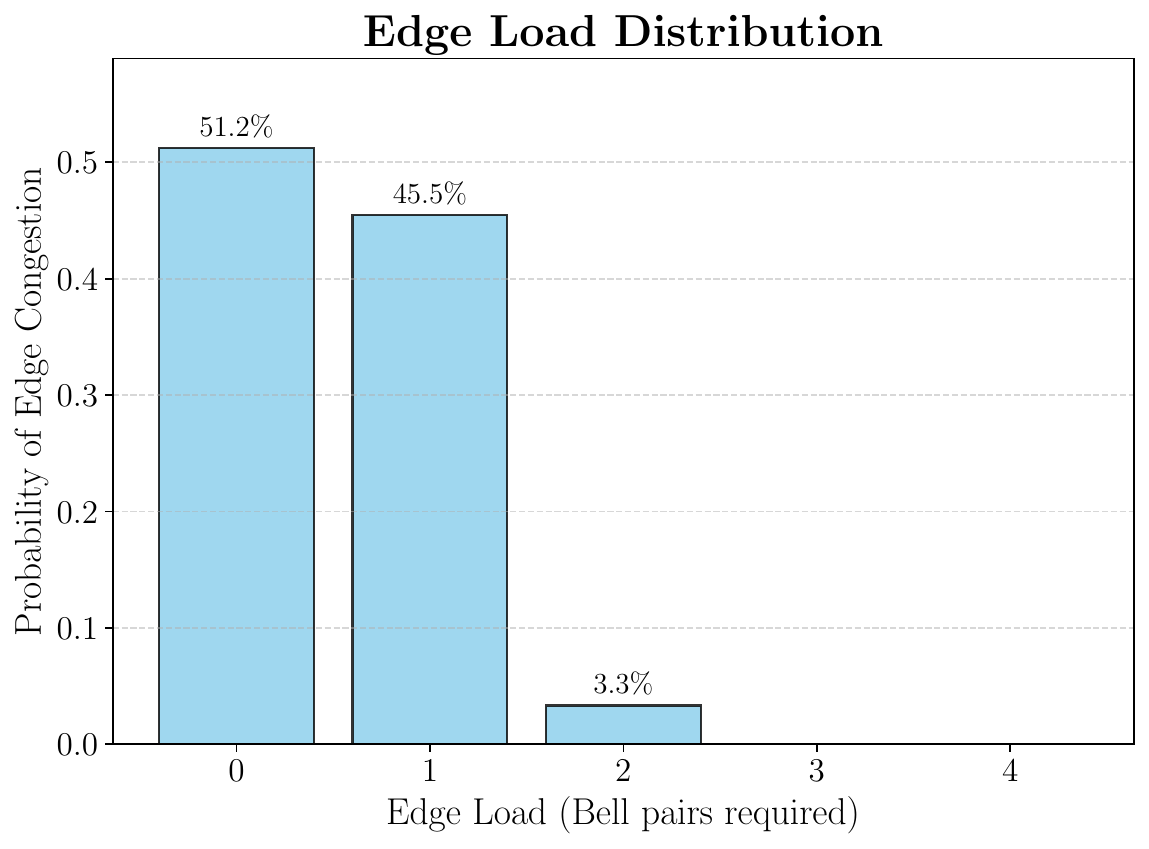}}
    \caption{\label{fig:fig3}  Edge load distribution. Probability distribution of edge loads under a maximal switching request ($n/2$ simultaneous pairs) obtained via the multi-path routing protocol for $n=128$. 
}
\end{figure}

Finally, we address the scalability of the proposed architecture by analyzing the resource overhead required to resolve conflicts as the network size grows. While theoretical results guarantee 2-rearrangeability for hypercubes, it is crucial to verify that the congestion does not diverge when employing realistic routing protocols. Fig. \ref{fig:fig4} displays the maximum edge load, defined as the maximum number of Bell pairs required on the most congested link, under full saturation traffic ($n/2$ pairs) for system sizes ranging from $n=2^3$ to $n=2^{10}$. We observe that the worst-case congestion saturates at a small constant value ($R_{\text{req}} \le 4$). The average maximum load increases gradually with the network size but remains strictly bounded by this worst-case limit. The slight increase to 4 in the worst case is due to the non-optimized heuristic algorithm employed. Crucially, the absence of significant divergence confirms that the nested quantum switch maintains its efficient $O(n \log n)$ total resource scaling, making it viable for large-scale quantum internet implementations.

\begin{figure}
    {\includegraphics[width=0.95\columnwidth]{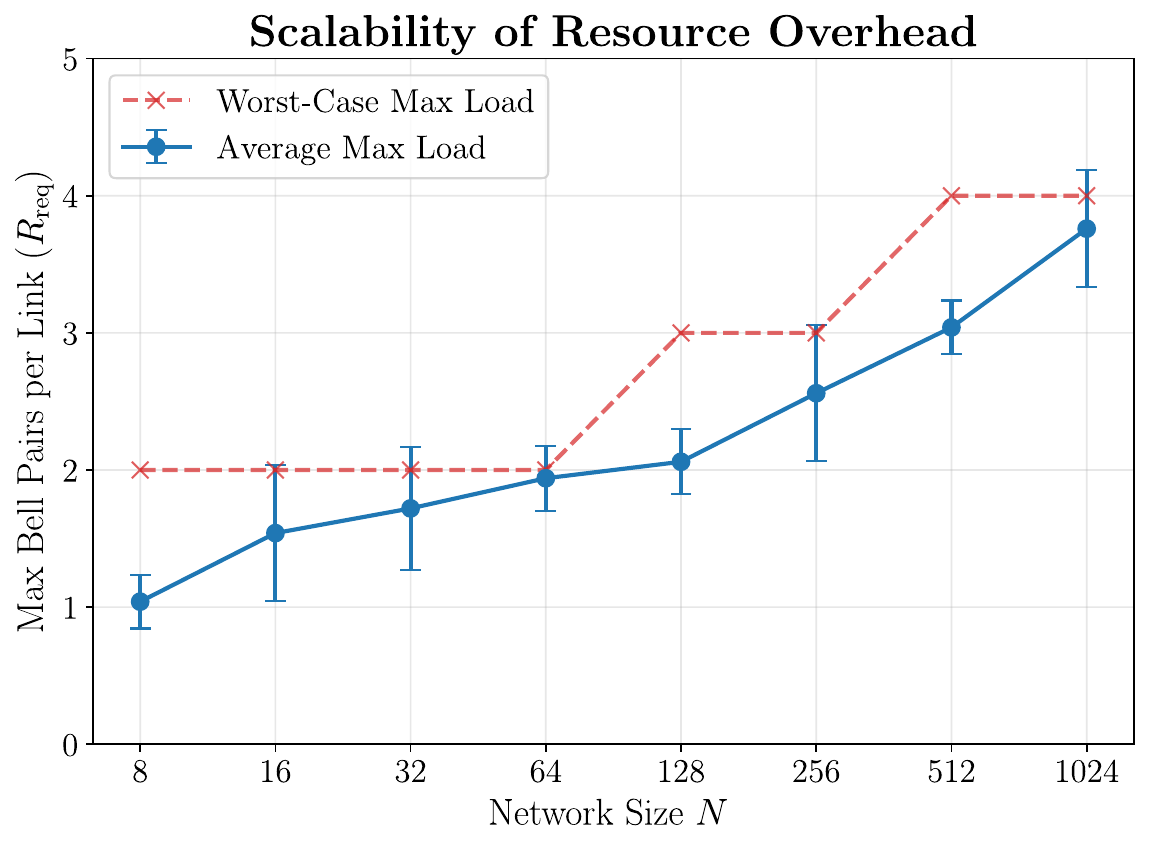}}
    \caption{\label{fig:fig4}  Scalability of resource overhead. The maximum number of Bell pairs per link ($R_{\text{req}}$) required to satisfy a full permutation request is plotted against the network size $n$. The blue dots represent the average maximum load over different random permutations, while the red crosses indicate the absolute worst-case congestion observed. 
}
\end{figure}

\subsection{Fidelity analysis}
To evaluate the physical viability of the nested switch, we provide some analytical insight on the end-to-end fidelity analysis by modeling the generation of an end-to-end connection across $L$ hops under realistic noise. We assume the initial shared resources are imperfect Bell pairs, modeled as Werner states \cite{nielsen_chuang_2010} with depolarizing parameter $p_0$:
\begin{equation}
\rho_0 = p_0 |\Phi^+\rangle\langle\Phi^+| + \frac{1-p_0}{4} \mathbb{I}_{4}.
\end{equation}
The fidelity of each elementary link is $F_0 = (3p_0 + 1)/4$. To establish a connection between the source and target nodes, $L-1$ entanglement swapping operations must be performed along the routed path. We consider two primary sources of operational noise, imperfect Swapping, i.e., each Bell State Measurement is modeled as an ideal projection followed by a depolarizing channel with parameter $p_{\text{swap}}$; and memory, i.e., storing the qubits while the classical routing information propagates induces decoherence. We model this as a local depolarizing channel with parameter $p_{\text{mem}} = e^{-t/T}$, where $t$ is the storage time and $T$ is the effective coherence time. 

When swapping two links characterized by parameters $p_1$ and $p_2$, the resulting state remains a Werner state with an updated parameter $p_{\text{new}} = p_1 p_2 p_{\text{swap}} p_{\text{mem}}^2$. Iterating this recurrence relation over the $L$ edges of the path yields the depolarizing parameter for the final end-to-end state:
\begin{equation}
p_L = p_0^L \left( p_{\text{swap}} p_{\text{mem}}^2 \right)^{L-1}.
\end{equation}
The corresponding end-to-end fidelity is therefore:
\begin{equation}
F_L = \frac{1}{4} \left( 1 + 3 p_0^L \left( p_{\text{swap}} p_{\text{mem}}^2 \right)^{L-1} \right).
\end{equation}
As stressed before, the nested architecture maintains a mean path length of around $L \approx (\log n)/2$, even under substantial node failures (see Fig. \ref{fig:fig1}). Substituting this logarithmic depth into the fidelity equation demonstrates that the end-to-end fidelity decays only \textit{polynomially} with the network size $n$. This guarantees that high-fidelity connections can be established across the network, even with losses,  without requiring complex purification protocols at intermediate nodes.

\section{Graph-state variant} \label{sec:graphstate}
We briefly discuss a conceptual variant of the nested quantum switch in which, instead of treating Bell pairs as independent routing resources, all qubits stored at each node are merged into a single qubit, forming a graph state with an edge between any two nodes previously connected by a Bell state.k. This construction has total qubit overhead $O(n)$ and may be viewed as a global entanglement resource from which Bell pairs are extracted by local measurements. The goal of this variant is not to realize arbitrary connection patterns, but to explore how far a highly merged resource can be pushed.

In the nested merged variant, to generate an end-to-end Bell pair along a path, intermediate vertices and their neighboring ones must be measured out in order to isolate the desired entanglement.  Pauli $Z$ measurements are used to isolate the path and Pauli $X/Y$ measurements to propagate the entanglement through it. A single routed connection, therefore, consumes not only the qubits along the path, but also effectively blocks a neighborhood of size $O(d)$ around each intermediate node. Since a typical path length is $O(d)$, the total excluded region per connection scales as $O(d^2) = O(\log^2 n)$. As a consequence, even under ideal scheduling, the number of Bell pairs that can be generated simultaneously from the global graph-state resource scales as  $O\!\left({n}/{\log^2 n}\right).$

We further evaluate the performance of this approach using a similar path-finding algorithm as Sec. \ref{sec:performance}, accounting for the effective qubit loss after each connection is generated. In this variant, establishing a connection necessitates measuring out all qubits along the chosen path and all immediate neighbors, thereby removing these qubits from the available resource. 

\begin{figure}
    {\includegraphics[width=0.95\columnwidth]{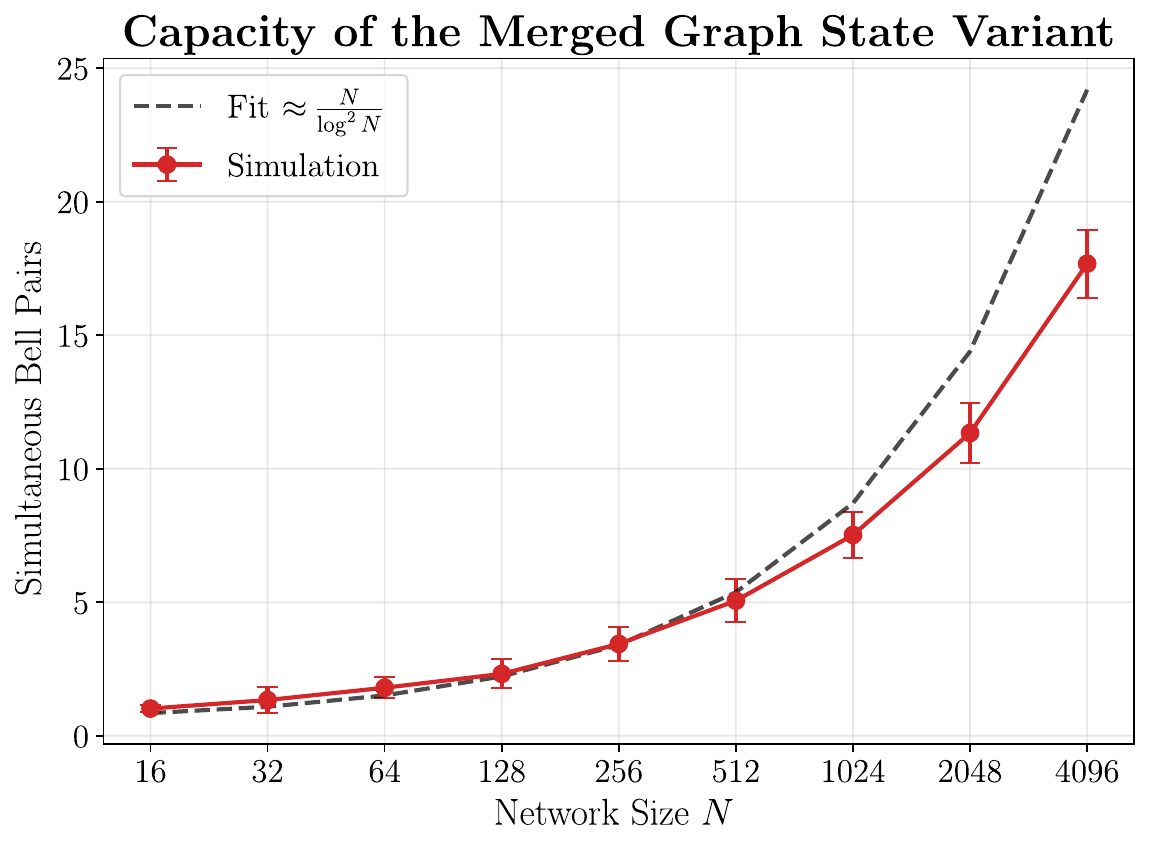}}
    \caption{\label{fig:fig5} Capacity of the merged graph-state variant. The maximum number of simultaneous Bell pairs $S$ that can be extracted from a unified $n$-qubit nested graph state. Red dots indicate results from a path-finding simulation where each connection effectively "consumes" its path and neighboring qubits. The dashed line represents the theoretical fit $S \approx n / \log^2 n$. 
}
\end{figure}

Fig. \ref{fig:fig5} displays the number of simultaneous Bell pairs obtained via our path-finding simulation. While the data exhibits a trend compatible with the theoretical $O(n/\log^2 n)$ scaling, we observe a growing downward deviation for larger networks. This reflects the limitations of our non-optimized routing algorithm, which increasingly struggles with path collisions induced by the high density of excluded regions in the merged-resource setting.

It is instructive to contrast this behavior with standard 2D cluster states  \cite{Freund2024,Asavanant2019-ks}, where the number of simultaneous long-range connections is fundamentally bottlenecked by the lattice geometry and scales as $O(\sqrt{n})$. Because $\sqrt{n}$ grows faster than $n/\log^2 n$ for small $n$, a 2D lattice actually provides a higher absolute capacity for finite-size networks within near-term reach. However, the $O(n/\log^2 n)$ throughput scaling of the merged  variant is asymptotically superior, eventually surpassing the $O(\sqrt{n})$ bound at large network sizes. Crucially the unmerged and independent nested switch model introduced in Sec. \ref{sec:model} removes this penalty entirely, restoring the full $O(n)$ parallel capacity and definitively overcoming the dimensionality limitations of standard 2D architectures.


The resource consumption directly relates to the end-to-end fidelity of the generated connections. Unlike the independent Bell-pair architecture, where noise scales with the linear path length $O(\log n)$, extracting a target pair from a merged graph state requires local projective measurements on all qubits within the excluded neighborhood. Assuming the global graph state is subject to local depolarizing noise, or equivalently, that the local measurement operations  are imperfect, we can assign an operational fidelity $F_0$ to each node. The total number of measured qubits required to isolate a single path is proportional to $(\log n)^2$, where logical errors accumulate multiplicatively with each measured vertex in the neighborhood \cite{MorRuiz2024}. Consequently, the final fidelity of the extracted Bell pair scales exponentially with the excluded region as
\begin{equation}
    F_{\text{end}} \propto F_0^{O((\log n)^2)}.
\end{equation}
While this polynomial scaling exhibits a higher exponent than the $O(\log n)$ path length of the independent Bell-pair switch, this variant offers significant physical advantages. Because the graph state can act as a pre-shared resource, it can be probabilistically generated and purified offline prior to any connectivity request. When a request is made, the connections are extracted on-demand using only fast single-qubit measurements, overcoming the accumulation of quantum memory errors and avoiding the use of noisy two-qubit swapping operations, potentially allowing for higher fidelities.

\section{Conclusions} \label{sec:discussion}
In this work, we have introduced the \textit{nested quantum switch}, a fully distributed network architecture that enables universal bipartite simultaneous connectivity without relying on a central hub. By structuring pre-shared entanglement according to a projected hypercube geometry, we achieve a resource complexity of $O(n \log n)$ total Bell pairs, requiring only logarithmic memory per node. This offers a scalable alternative to existing distributed schemes that demand $O(n^2)$ resources, while avoiding the single-point-of-failure risks associated with centralized star topologies.

We have formally shown that this architecture supports arbitrary simultaneous pairing requests, proving theoretical universality. Our numerical analysis demonstrates how simple algorithms can reach such performance and that the nested switch exhibits practical universality and inherent resilience to network noise. Under random node failures, the system does not suffer a catastrophic breakdown but instead degrades gradually, allowing a large fraction of connections to be served even when significant portions of the network are lost. Furthermore, the entanglement cost per generated pair remains stable under these failure conditions, indicating that the topological advantages of the nested switch persist in realistic, imperfect scenarios.

While we have discussed the usage of such a universal quantum switch in a network setting, we point out that it can also be utilized in other contexts. This includes, e.g., the realization of remote two-qubit gates in quantum computational architectures \cite{Dr2025}.

Finally, we discuss and prove a low-memory variant based on graph states, which trades full universality for extreme memory efficiency, i.e., $O(1)$ qubits per node. Collectively, these results suggest that this hierarchical nested geometry provides a robust and efficient alternative for the next generation of scalable quantum networks. Future work may address the physical implementation of such connectivities and the integration of purification protocols to tackle link-level noise.

\section*{Acknowledgments}
This research was funded in whole or in part by the Austrian Science Fund (FWF) 10.55776/P36009,  10.55776/P36010, 10.55776/PAT1710825 and 10.55776/COE1. For open access purposes, the author has applied a CC BY public copyright license to any author accepted manuscript version arising from this submission. Finanziert von der Europ\"aischen Union - NextGenerationEU.

\bibliographystyle{apsrev4-2}
\bibliography{Qswitch.bib}

@article{Vardoyan2023,
  title = {On the Capacity Region of Bipartite and Tripartite Entanglement Switching},
  volume = {8},
  ISSN = {2376-3647},
  url = {http://dx.doi.org/10.1145/3571809},
  DOI = {10.1145/3571809},
  number = {1–2},
  journal = {ACM Transactions on Modeling and Performance Evaluation of Computing Systems},
  publisher = {Association for Computing Machinery (ACM)},
  author = {Vardoyan,  Gayane and Nain,  Philippe and Guha,  Saikat and Towsley,  Don},
  year = {2023},
  month = mar,
  pages = {1–18}
}

@article{Hein2004,
  title = {Multiparty entanglement in graph states},
  author = {Hein, M. and Eisert, J. and Briegel, H. J.},
  journal = {Phys. Rev. A},
  volume = {69},
  issue = {6},
  pages = {062311},
  numpages = {20},
  year = {2004},
  month = {Jun},
  publisher = {American Physical Society},
  doi = {10.1103/PhysRevA.69.062311},
  url = {https://link.aps.org/doi/10.1103/PhysRevA.69.062311}
}

@article{Hein2006,
author = {Hein, M. and Dür, W. and Eisert, Jens and Raussendorf, Robert and Nest, M. and Briegel, H.},
year = {2006},
month = {03},
journal={in Quantum Computers, Algorithms and Chaos,
Proceedings of the International School of Physics “Enrico Fermi,”
Vol. 162, Varenna, 2005, edited by G. Casati, D. L. Shepelyansky,
P. Zoller, and G. Benenti (IOS Press, Amsterdam},
title = {Entanglement in Graph States and its Applications},
volume = {162},
pages={115-218},
  url = {https://doi.org/10.3254/978-1-61499-018-5-115},
  DOI = {10.3254/978-1-61499-018-5-115}
}

@article{MiguelRamiro2021,
  title = {Genuine quantum networks with superposed tasks and addressing},
  volume = {7},
  ISSN = {2056-6387},
  url = {http://dx.doi.org/10.1038/s41534-021-00472-5},
  DOI = {10.1038/s41534-021-00472-5},
  number = {1},
  journal = {npj Quantum Information},
  publisher = {Springer Science and Business Media LLC},
  author = {Miguel-Ramiro,  J. and Pirker,  A. and D\"{u}r,  W.},
  year = {2021},
  pages ={135},
  month = sep 
}

@article{Vardoyan2019,
  title = {On the Stochastic Analysis of a Quantum Entanglement Switch},
  volume = {47},
  ISSN = {0163-5999},
  url = {http://dx.doi.org/10.1145/3374888.3374899},
  DOI = {10.1145/3374888.3374899},
  number = {2},
  journal = {ACM SIGMETRICS Performance Evaluation Review},
  publisher = {Association for Computing Machinery (ACM)},
  author = {Vardoyan,  Gayane and Guha,  Saikat and Nain,  Philippe and Towsley,  Don},
  year = {2019},
  month = dec,
  pages = {27–29}
}

@article{Vardoyan2021,
  title = {On the exact analysis of an idealized quantum switch},
  volume = {144},
  ISSN = {0166-5316},
  url = {http://dx.doi.org/10.1016/j.peva.2020.102141},
  DOI = {10.1016/j.peva.2020.102141},
  journal = {Performance Evaluation},
  publisher = {Elsevier BV},
  author = {Vardoyan,  Gayane and Guha,  Saikat and Nain,  Philippe and Towsley,  Don},
  year = {2020},
  month = dec,
  pages = {102141}
}

@misc{Dai2021,
  doi = {10.48550/ARXIV.2110.04116},
  url = {https://arxiv.org/abs/2110.04116},
  author = {Dai,  Wenhan and Rinaldi,  Anthony and Towsley,  Don},
  title = {Entanglement Swapping in Quantum Switches: Protocol Design and Stability Analysis},
  publisher = {arXiv},
  year = {2021},
  copyright = {Creative Commons Attribution Non Commercial No Derivatives 4.0 International}
}

@misc{MR2025,
  doi = {10.48550/ARXIV.2508.03806},
  url = {https://arxiv.org/abs/2508.03806},
  author = {Miguel-Ramiro,  Jorge and Illiano,  Jessica and Mazza,  Francesco and Pirker,  Alexander and Freund,  Julia and Cacciapuoti,  Angela Sara and Caleffi,  Marcello and D\"{u}r,  Wolfgang},
  title = {QPing: a Quantum Ping Primitive for Quantum Networks},
  publisher = {arXiv},
  year = {2025},
  copyright = {arXiv.org perpetual,  non-exclusive license}
}

@misc{Mazza2025,
  doi = {10.48550/ARXIV.2510.15776},
  url = {https://arxiv.org/abs/2510.15776},
  author = {Mazza,  Francesco and Miguel-Ramiro,  Jorge and Illiano,  Jessica and Pirker,  Alexander and Caleffi,  Marcello and Cacciapuoti,  Angela Sara and D\"{u}r,  Wolfgang},
  title = {Flexible Qubit Allocation of Network Resource States},
  publisher = {arXiv},
  year = {2025},
  copyright = {Creative Commons Attribution Non Commercial No Derivatives 4.0 International}
}

@article{MorRuiz2024,
  title = {Influence of Noise in Entanglement-Based Quantum Networks},
  volume = {42},
  ISSN = {1558-0008},
  url = {http://dx.doi.org/10.1109/JSAC.2024.3380089},
  DOI = {10.1109/jsac.2024.3380089},
  number = {7},
  journal = {IEEE Journal on Selected Areas in Communications},
  publisher = {Institute of Electrical and Electronics Engineers (IEEE)},
  author = {Mor-Ruiz,  Maria Flors and D\"{u}r,  Wolfgang},
  year = {2024},
  month = jul,
  pages = {1793–1807}
}

@inproceedings{Panigrahy2023,
  title = {On the Capacity Region of a Quantum Switch with Entanglement Purification},
  url = {http://dx.doi.org/10.1109/INFOCOM53939.2023.10229003},
  DOI = {10.1109/infocom53939.2023.10229003},
  booktitle = {IEEE INFOCOM 2023 - IEEE Conference on Computer Communications},
  publisher = {IEEE},
  author = {Panigrahy,  Nitish K. and Vasantam,  Thirupathaiah and Towsley,  Don and Tassiulas,  Leandros},
  year = {2023},
  month = may,
  pages = {1–10}
}

@inproceedings{Kumar2023,
  title = {Optimal Entanglement Distillation Policies for Quantum Switches},
  url = {http://dx.doi.org/10.1109/QCE57702.2023.00135},
  DOI = {10.1109/qce57702.2023.00135},
  booktitle = {2023 IEEE International Conference on Quantum Computing and Engineering (QCE)},
  publisher = {IEEE},
  author = {Kumar,  Vivek and Chandra,  Nitish K. and Seshadreesan,  Kaushik P. and Scheller-Wolf,  Alan and Tayur,  Sridhar},
  year = {2023},
  month = sep,
  pages = {1198–1204}
}

@article{Avis2023,
  title = {Analysis of multipartite entanglement distribution using a central quantum-network node},
  author = {Avis, Guus and Rozp\ifmmode \mbox{\k{e}}\else \k{e}\fi{}dek, Filip and Wehner, Stephanie},
  journal = {Phys. Rev. A},
  volume = {107},
  issue = {1},
  pages = {012609},
  numpages = {36},
  year = {2023},
  month = {Jan},
  publisher = {American Physical Society},
  doi = {10.1103/PhysRevA.107.012609},
  url = {https://link.aps.org/doi/10.1103/PhysRevA.107.012609}
}

@inproceedings{Gauthier2023,
  series = {QuNet ’23},
  title = {A Control Architecture for Entanglement Generation Switches in Quantum Networks},
  url = {http://dx.doi.org/10.1145/3610251.3610552},
  DOI = {10.1145/3610251.3610552},
  booktitle = {Proceedings of the 1st Workshop on Quantum Networks and Distributed Quantum Computing},
  publisher = {ACM},
  author = {Gauthier,  Scarlett and Vardoyan,  Gayane and Wehner,  Stephanie},
  year = {2023},
  month = sep,
  pages = {38–44},
  collection = {QuNet ’23}
}

@article{Lee2022,
  title = {A quantum router architecture for high-fidelity entanglement flows in quantum networks},
  volume = {8},
  ISSN = {2056-6387},
  url = {http://dx.doi.org/10.1038/s41534-022-00582-8},
  DOI = {10.1038/s41534-022-00582-8},
  number = {1},
  journal = {npj Quantum Information},
  publisher = {Springer Science and Business Media LLC},
  author = {Lee,  Yuan and Bersin,  Eric and Dahlberg,  Axel and Wehner,  Stephanie and Englund,  Dirk},
  year = {2022},
  pages = {75},
  month = jun 
}

@article{Pant2019,
  title = {Routing entanglement in the quantum internet},
  volume = {5},
  ISSN = {2056-6387},
  url = {http://dx.doi.org/10.1038/s41534-019-0139-x},
  DOI = {10.1038/s41534-019-0139-x},
  number = {1},
  journal = {npj Quantum Information},
  publisher = {Springer Science and Business Media LLC},
  author = {Pant,  Mihir and Krovi,  Hari and Towsley,  Don and Tassiulas,  Leandros and Jiang,  Liang and Basu,  Prithwish and Englund,  Dirk and Guha,  Saikat},
  year = {2019},
pages={25},
  month = mar 
}

@article{Sutcliffe2023,
  title = {Multiuser Entanglement Distribution in Quantum Networks Using Multipath Routing},
  volume = {4},
  ISSN = {2689-1808},
  url = {http://dx.doi.org/10.1109/TQE.2023.3329714},
  DOI = {10.1109/tqe.2023.3329714},
  journal = {IEEE Transactions on Quantum Engineering},
  publisher = {Institute of Electrical and Electronics Engineers (IEEE)},
  author = {Sutcliffe,  Evan and Beghelli,  Alejandra},
  year = {2023},
  pages = {1–15}
}

@article{Gyongyosi2019,
  title = {Opportunistic Entanglement Distribution for the Quantum Internet},
  volume = {9},
  ISSN = {2045-2322},
  url = {http://dx.doi.org/10.1038/s41598-019-38495-w},
  DOI = {10.1038/s41598-019-38495-w},
  number = {1},
  journal = {Scientific Reports},
  publisher = {Springer Science and Business Media LLC},
  author = {Gyongyosi,  Laszlo and Imre,  Sandor},
  year = {2019},
pages={2219},
  month = feb 
}

@misc{Abane2024,
  doi = {10.48550/ARXIV.2408.01234},
  url = {https://arxiv.org/abs/2408.01234},
  author = {Abane,  Amar and Cubeddu,  Michael and Mai,  Van Sy and Battou,  Abdella},
  title = {Entanglement Routing in Quantum Networks: A Comprehensive Survey},
  publisher = {arXiv},
  year = {2024},
  copyright = {Creative Commons Attribution 4.0 International}
}

@article{Shi2024,
  title = {Concurrent Entanglement Routing for Quantum Networks: Model and Designs},
  volume = {32},
  ISSN = {1558-2566},
  url = {http://dx.doi.org/10.1109/TNET.2023.3343748},
  DOI = {10.1109/tnet.2023.3343748},
  number = {3},
  journal = {IEEE/ACM Transactions on Networking},
  publisher = {Institute of Electrical and Electronics Engineers (IEEE)},
  author = {Shi,  Shouqian and Zhang,  Xiaoxue and Qian,  Chen},
  year = {2024},
  month = jun,
  pages = {2205–2220}
}

@article{Pirker2018,
  title = {Modular architectures for quantum networks},
  volume = {20},
  ISSN = {1367-2630},
  url = {http://dx.doi.org/10.1088/1367-2630/aac2aa},
  DOI = {10.1088/1367-2630/aac2aa},
  number = {5},
  journal = {New Journal of Physics},
  publisher = {IOP Publishing},
  author = {Pirker,  A and Walln\"{o}fer,  J and D\"{u}r,  W},
  year = {2018},
  month = may,
  pages = {053054}
}

@article{Pirker2019,
  title = {A quantum network stack and protocols for reliable entanglement-based networks},
  volume = {21},
  ISSN = {1367-2630},
  url = {http://dx.doi.org/10.1088/1367-2630/ab05f7},
  DOI = {10.1088/1367-2630/ab05f7},
  number = {3},
  journal = {New Journal of Physics},
  publisher = {IOP Publishing},
  author = {Pirker,  A and D\"{u}r,  W},
  year = {2019},
  month = mar,
  pages = {033003}
}

@article{MiguelRamiro2023,
  title = {Optimized Quantum Networks},
  volume = {7},
  ISSN = {2521-327X},
  url = {http://dx.doi.org/10.22331/q-2023-02-09-919},
  DOI = {10.22331/q-2023-02-09-919},
  journal = {Quantum},
  publisher = {Verein zur Forderung des Open Access Publizierens in den Quantenwissenschaften},
  author = {Miguel-Ramiro,  Jorge and Pirker,  Alexander and D\"{u}r,  Wolfgang},
  year = {2023},
  month = feb,
  pages = {919}
}

@article{Hahn2019,
  title = {Quantum network routing and local complementation},
  volume = {5},
  ISSN = {2056-6387},
  url = {http://dx.doi.org/10.1038/s41534-019-0191-6},
  DOI = {10.1038/s41534-019-0191-6},
  number = {1},
  journal = {npj Quantum Information},
  publisher = {Springer Science and Business Media LLC},
  author = {Hahn,  F. and Pappa,  A. and Eisert,  J.},
  year = {2019},
pages ={76},
  month = sep 
}

@article{Meignant2019,
  title = {Distributing graph states over arbitrary quantum networks},
  author = {Meignant, Cl\'ement and Markham, Damian and Grosshans, Fr\'ed\'eric},
  journal = {Phys. Rev. A},
  volume = {100},
  issue = {5},
  pages = {052333},
  numpages = {6},
  year = {2019},
  month = {Nov},
  publisher = {American Physical Society},
  doi = {10.1103/PhysRevA.100.052333},
  url = {https://link.aps.org/doi/10.1103/PhysRevA.100.052333}
}

@inproceedings{Fischer2021,
  title = {Distributing Graph States Across Quantum Networks},
  url = {http://dx.doi.org/10.1109/QCE52317.2021.00049},
  DOI = {10.1109/qce52317.2021.00049},
  booktitle = {2021 IEEE International Conference on Quantum Computing and Engineering (QCE)},
  publisher = {IEEE},
  author = {Fischer,  Alex and Towsley,  Don},
  year = {2021},
  month = oct,
  pages = {324–333}
}

@article{Sprague1994,
  title = {Routings for involutions of a hypercube},
  volume = {48},
  ISSN = {0166-218X},
  url = {http://dx.doi.org/10.1016/0166-218X(92)00126-7},
  DOI = {10.1016/0166-218x(92)00126-7},
  number = {2},
  journal = {Discrete Applied Mathematics},
  publisher = {Elsevier BV},
  author = {Sprague,  Alan P. and Tamaki,  Hisao},
  year = {1994},
  month = jan,
  pages = {175–186}
}

@article{Harary1988,
  title = {A survey of the theory of hypercube graphs},
  volume = {15},
  ISSN = {0898-1221},
  url = {http://dx.doi.org/10.1016/0898-1221(88)90213-1},
  DOI = {10.1016/0898-1221(88)90213-1},
  number = {4},
  journal = {Computers \& Mathematics with Applications},
  publisher = {Elsevier BV},
  author = {Harary,  Frank and Hayes,  John P. and Wu,  Horng-Jyh},
  year = {1988},
  pages = {277–289}
}

@article{Fink2007,
  title = {Perfect matchings extend to Hamilton cycles in hypercubes},
  volume = {97},
  ISSN = {0095-8956},
  url = {http://dx.doi.org/10.1016/j.jctb.2007.02.007},
  DOI = {10.1016/j.jctb.2007.02.007},
  number = {6},
  journal = {Journal of Combinatorial Theory,  Series B},
  publisher = {Elsevier BV},
  author = {Fink,  Jiří},
  year = {2007},
  month = nov,
  pages = {1074–1076}
}

@book{Diestel2025,
  title = {Graph Theory},
  ISBN = {9783662701072},
  ISSN = {2197-5612},
  url = {http://dx.doi.org/10.1007/978-3-662-70107-2},
  DOI = {10.1007/978-3-662-70107-2},
  journal = {Graduate Texts in Mathematics},
  publisher = {Springer Berlin Heidelberg},
  author = {Diestel,  Reinhard},
  year = {2025}
}

@article{Gu1997,
  title = {Routing a Permutation in the Hypercube by Two Sets of Edge Disjoint Paths},
  volume = {44},
  ISSN = {0743-7315},
  url = {http://dx.doi.org/10.1006/jpdc.1997.1358},
  DOI = {10.1006/jpdc.1997.1358},
  number = {2},
  journal = {Journal of Parallel and Distributed Computing},
  publisher = {Elsevier BV},
  author = {Gu,  Qian-Ping and Tamaki,  Hisao},
  year = {1997},
  month = aug,
  pages = {147–152}
}

@article{Choi1993,
  title = {Rearrangeable Circuit-Switched Hypercube Architectures for Routing Permutations},
  volume = {19},
  ISSN = {0743-7315},
  url = {http://dx.doi.org/10.1006/jpdc.1993.1097},
  DOI = {10.1006/jpdc.1993.1097},
  number = {2},
  journal = {Journal of Parallel and Distributed Computing},
  publisher = {Elsevier BV},
  author = {Choi,  S.B. and Somani,  A.K.},
  year = {1993},
  month = oct,
  pages = {125–130}
}

@article{Shen1994,
  title = {Realization of an arbitrary permutation on a hypercube},
  volume = {51},
  ISSN = {0020-0190},
  url = {http://dx.doi.org/10.1016/0020-0190(94)90002-7},
  DOI = {10.1016/0020-0190(94)90002-7},
  number = {5},
  journal = {Information Processing Letters},
  publisher = {Elsevier BV},
  author = {Shen,  Xiaojun and Hu,  Qing and Liang,  Weifa},
  year = {1994},
  month = sep,
  pages = {237–243}
}

@article{PrezCastro2024,
  title = {Simulation of Fidelity in Entanglement-Based Networks with Repeater Chains},
  volume = {14},
  ISSN = {2076-3417},
  url = {http://dx.doi.org/10.3390/app142311270},
  DOI = {10.3390/app142311270},
  number = {23},
  journal = {Applied Sciences},
  publisher = {MDPI AG},
  author = {Pérez Castro,  David and Fernández Vilas,  Ana and Fernández Veiga,  Manuel and Blanco Rodríguez,  Mateo and Díaz Redondo,  Rebeca P.},
  year = {2024},
  month = dec,
  pages = {11270}
}

@article{Wei2022,
  title = {Towards Real‐World Quantum Networks: A Review},
  volume = {16},
  ISSN = {1863-8899},
  url = {http://dx.doi.org/10.1002/lpor.202100219},
  DOI = {10.1002/lpor.202100219},
  number = {3},
  journal = {Laser \& Photonics Reviews},
  publisher = {Wiley},
  author = {Wei,  Shi‐Hai and Jing,  Bo and Zhang,  Xue‐Ying and Liao,  Jin‐Yu and Yuan,  Chen‐Zhi and Fan,  Bo‐Yu and Lyu,  Chen and Zhou,  Dian‐Li and Wang,  You and Deng,  Guang‐Wei and Song,  Hai‐Zhi and Oblak,  Daniel and Guo,  Guang‐Can and Zhou,  Qiang},
  year = {2022},
pages={2100219},
  month = jan 
}

@article{Azuma2023,
  title = {Quantum repeaters: From quantum networks to the quantum internet},
  author = {Azuma, Koji and Economou, Sophia E. and Elkouss, David and Hilaire, Paul and Jiang, Liang and Lo, Hoi-Kwong and Tzitrin, Ilan},
  journal = {Rev. Mod. Phys.},
  volume = {95},
  issue = {4},
  pages = {045006},
  numpages = {66},
  year = {2023},
  month = {Dec},
  publisher = {American Physical Society},
  doi = {10.1103/RevModPhys.95.045006},
  url = {https://link.aps.org/doi/10.1103/RevModPhys.95.045006}
}

@inproceedings{Kozlowski2019,
  series = {NANOCOM ’19},
  title = {Towards Large-Scale Quantum Networks},
  url = {http://dx.doi.org/10.1145/3345312.3345497},
  DOI = {10.1145/3345312.3345497},
  booktitle = {Proceedings of the Sixth Annual ACM International Conference on Nanoscale Computing and Communication},
  publisher = {ACM},
  author = {Kozlowski,  Wojciech and Wehner,  Stephanie},
  year = {2019},
  month = sep,
  pages = {1–7},
  collection = {NANOCOM ’19}
}

@article{Wehner2018,
  title = {Quantum internet: A vision for the road ahead},
  volume = {362},
  ISSN = {1095-9203},
  url = {http://dx.doi.org/10.1126/science.aam9288},
  DOI = {10.1126/science.aam9288},
  number = {6412},
  journal = {Science},
  publisher = {American Association for the Advancement of Science (AAAS)},
  author = {Wehner,  Stephanie and Elkouss,  David and Hanson,  Ronald},
  year = {2018},
pages = {eaam9288},
  month = oct 
}

@article{Kimble2008,
  title = {The quantum internet},
  volume = {453},
  ISSN = {1476-4687},
  url = {http://dx.doi.org/10.1038/nature07127},
  DOI = {10.1038/nature07127},
  number = {7198},
  journal = {Nature},
  publisher = {Springer Science and Business Media LLC},
  author = {Kimble,  H. J.},
  year = {2008},
  month = jun,
  pages = {1023–1030}
}

@book{Wolf2021,
  title = {Quantum Key Distribution: An Introduction with Exercises},
  ISBN = {9783030739911},
  ISSN = {1616-6361},
  url = {http://dx.doi.org/10.1007/978-3-030-73991-1},
  DOI = {10.1007/978-3-030-73991-1},
  journal = {Lecture Notes in Physics},
  publisher = {Springer International Publishing},
  author = {Wolf,  Ramona},
  year = {2021}
}

@article{Cao2022,
  title = {The Evolution of Quantum Key Distribution Networks: On the Road to the Qinternet},
  volume = {24},
  ISSN = {2373-745X},
  url = {http://dx.doi.org/10.1109/COMST.2022.3144219},
  DOI = {10.1109/comst.2022.3144219},
  number = {2},
  journal = {IEEE Communications Surveys \& Tutorials},
  publisher = {Institute of Electrical and Electronics Engineers (IEEE)},
  author = {Cao,  Yuan and Zhao,  Yongli and Wang,  Qin and Zhang,  Jie and Ng,  Soon Xin and Hanzo,  Lajos},
  year = {2022},
  pages = {839–894}
}

@article{Mehic2020,
  title = {Quantum Key Distribution: A Networking Perspective},
  volume = {53},
  ISSN = {1557-7341},
  url = {http://dx.doi.org/10.1145/3402192},
  DOI = {10.1145/3402192},
  number = {5},
  journal = {ACM Computing Surveys},
  publisher = {Association for Computing Machinery (ACM)},
  author = {Mehic,  Miralem and Niemiec,  Marcin and Rass,  Stefan and Ma,  Jiajun and Peev,  Momtchil and Aguado,  Alejandro and Martin,  Vicente and Schauer,  Stefan and Poppe,  Andreas and Pacher,  Christoph and Voznak,  Miroslav},
  year = {2020},
  month = sep,
  pages = {1–41}
}

@article{Sekatski2020,
  title = {Optimal distributed sensing in noisy environments},
  author = {Sekatski, P. and W\"olk, S. and D\"ur, W.},
  journal = {Phys. Rev. Res.},
  volume = {2},
  issue = {2},
  pages = {023052},
  numpages = {8},
  year = {2020},
  month = {Apr},
  publisher = {American Physical Society},
  doi = {10.1103/PhysRevResearch.2.023052},
  url = {https://link.aps.org/doi/10.1103/PhysRevResearch.2.023052}
}

@article{Zhang2021,
  title = {Distributed quantum sensing},
  volume = {6},
  ISSN = {2058-9565},
  url = {http://dx.doi.org/10.1088/2058-9565/abd4c3},
  DOI = {10.1088/2058-9565/abd4c3},
  number = {4},
  journal = {Quantum Science and Technology},
  publisher = {IOP Publishing},
  author = {Zhang,  Zheshen and Zhuang,  Quntao},
  year = {2021},
  month = jul,
  pages = {043001}
}

@article{Bugalho2025,
  title = {Private and Robust States for Distributed Quantum Sensing},
  volume = {9},
  ISSN = {2521-327X},
  url = {http://dx.doi.org/10.22331/q-2025-01-15-1596},
  DOI = {10.22331/q-2025-01-15-1596},
  journal = {Quantum},
  publisher = {Verein zur Forderung des Open Access Publizierens in den Quantenwissenschaften},
  author = {Bugalho,  Luís and Hassani,  Majid and Omar,  Yasser and Markham,  Damian},
  year = {2025},
  month = jan,
  pages = {1596}
}

@article{Cacciapuoti2020,
  title = {Quantum Internet: Networking Challenges in Distributed Quantum Computing},
  volume = {34},
  ISSN = {1558-156X},
  url = {http://dx.doi.org/10.1109/MNET.001.1900092},
  DOI = {10.1109/mnet.001.1900092},
  number = {1},
  journal = {IEEE Network},
  publisher = {Institute of Electrical and Electronics Engineers (IEEE)},
  author = {Cacciapuoti,  Angela Sara and Caleffi,  Marcello and Tafuri,  Francesco and Cataliotti,  Francesco Saverio and Gherardini,  Stefano and Bianchi,  Giuseppe},
  year = {2020},
  month = jan,
  pages = {137–143}
}

@article{Main2025,
  title = {Distributed quantum computing across an optical network link},
  volume = {638},
  ISSN = {1476-4687},
  url = {http://dx.doi.org/10.1038/s41586-024-08404-x},
  DOI = {10.1038/s41586-024-08404-x},
  number = {8050},
  journal = {Nature},
  publisher = {Springer Science and Business Media LLC},
  author = {Main,  D. and Drmota,  P. and Nadlinger,  D. P. and Ainley,  E. M. and Agrawal,  A. and Nichol,  B. C. and Srinivas,  R. and Araneda,  G. and Lucas,  D. M.},
  year = {2025},
  month = feb,
  pages = {383–388}
}

@article{Fitzsimons2017,
  title = {Private quantum computation: an introduction to blind quantum computing and related protocols},
  volume = {3},
  ISSN = {2056-6387},
  url = {http://dx.doi.org/10.1038/s41534-017-0025-3},
  DOI = {10.1038/s41534-017-0025-3},
  number = {1},
  journal = {npj Quantum Information},
  publisher = {Springer Science and Business Media LLC},
  author = {Fitzsimons,  Joseph F.},
  year = {2017},
 pages = {23},
  month = jun 
}

@article{Barz2012,
  title = {Demonstration of Blind Quantum Computing},
  volume = {335},
  ISSN = {1095-9203},
  url = {http://dx.doi.org/10.1126/science.1214707},
  DOI = {10.1126/science.1214707},
  number = {6066},
  journal = {Science},
  publisher = {American Association for the Advancement of Science (AAAS)},
  author = {Barz,  Stefanie and Kashefi,  Elham and Broadbent,  Anne and Fitzsimons,  Joseph F. and Zeilinger,  Anton and Walther,  Philip},
  year = {2012},
  month = jan,
  pages = {303–308}
}

@article{Freund2024,
  title = {Flexible quantum data bus for quantum networks},
  author = {Freund, Julia and Pirker, Alexander and D\"ur, Wolfgang},
  journal = {Phys. Rev. Res.},
  volume = {6},
  issue = {3},
  pages = {033267},
  numpages = {11},
  year = {2024},
  month = {Sep},
  publisher = {American Physical Society},
  doi = {10.1103/PhysRevResearch.6.033267},
  url = {https://link.aps.org/doi/10.1103/PhysRevResearch.6.033267}
}

@misc{Schoute2016,
  doi = {10.48550/ARXIV.1610.05238},
  url = {https://arxiv.org/abs/1610.05238},
  author = {Schoute,  Eddie and Mancinska,  Laura and Islam,  Tanvirul and Kerenidis,  Iordanis and Wehner,  Stephanie},
  title = {Shortcuts to quantum network routing},
  publisher = {arXiv},
  year = {2016},
  copyright = {arXiv.org perpetual,  non-exclusive license}
}

@article{Devulapalli2024,
  title = {Quantum routing with teleportation},
  author = {Devulapalli, Dhruv and Schoute, Eddie and Bapat, Aniruddha and Childs, Andrew M. and Gorshkov, Alexey V.},
  journal = {Phys. Rev. Res.},
  volume = {6},
  issue = {3},
  pages = {033313},
  numpages = {12},
  year = {2024},
  month = {Sep},
  publisher = {American Physical Society},
  doi = {10.1103/PhysRevResearch.6.033313},
  url = {https://link.aps.org/doi/10.1103/PhysRevResearch.6.033313}
}

@ARTICLE{Asavanant2019-ks,
  title     = "Generation of time-domain-multiplexed two-dimensional cluster
               state",
  author    = "Asavanant, Warit and Shiozawa, Yu and Yokoyama, Shota and
               Charoensombutamon, Baramee and Emura, Hiroki and Alexander,
               Rafael N and Takeda, Shuntaro and Yoshikawa, Jun-Ichi and
               Menicucci, Nicolas C and Yonezawa, Hidehiro and Furusawa, Akira",
  journal   = "Science",
  publisher = "American Association for the Advancement of Science (AAAS)",
  volume    =  366,
  number    =  6463,
  pages     = "373--376",
  month     =  oct,
  year      =  2019,
  copyright = "http://www.sciencemag.org/about/science-licenses-journal-article-reuse",
  language  = "en"
}

@article{Dr2025,
  title = {Long-ranged gates in quantum computation architectures with limited connectivity},
  volume = {11},
  ISSN = {2058-9565},
  url = {http://dx.doi.org/10.1088/2058-9565/ae20b6},
  DOI = {10.1088/2058-9565/ae20b6},
  number = {1},
  journal = {Quantum Science and Technology},
  publisher = {IOP Publishing},
  author = {D\"{u}r,  Wolfgang},
  year = {2025},
  month = nov,
  pages = {015013}
}

@book{nielsen_chuang_2010, 
place={Cambridge}, 
title={\emph{Quantum Computation and Quantum Information: 10th Anniversary Edition}}, DOI={10.1017/CBO9780511976667}, 
publisher={Cambridge University Press}, 
author={Nielsen, Michael A. and Chuang, Isaac L.}, 
year={2010}}

\end{document}